\documentclass[sigconf]{acmart}
\usepackage{subfigure}
\usepackage{times}
\usepackage{latexsym}
\usepackage{microtype}
\usepackage{stfloats}
\usepackage{supertabular}
\usepackage{threeparttable}
\usepackage{xcolor}

\usepackage{amssymb}

\usepackage{pifont}
\newcommand{\cmark}{\ding{51}}%
\newcommand{\xmark}{\ding{55}}%

\AtBeginDocument{%
  \providecommand\BibTeX{{%
    \normalfont B\kern-0.5em{\scshape i\kern-0.25em b}\kern-0.8em\TeX}}}

\copyrightyear{2022}
\acmYear{2022}
\setcopyright{acmcopyright}
\acmConference[SIGIR '22]{ Proceedings of the 45th International ACM SIGIR Conference on Research and Development in Information Retrieval}{July 11--15, 2022}{Madrid, Spain.}
\acmBooktitle{Proceedings of the 45th International ACM SIGIR Conference on Research and Development in Information Retrieval (SIGIR '22), July 11--15, 2022, Madrid, Spain}
\acmPrice{15.00}
\acmISBN{978-1-4503-8732-3/22/07}
\acmDOI{10.1145/3477495.3531963}

\usepackage{multirow}

\begin{document}
\fancyhead{}
\begin{sloppypar}

\title{Decoupled Side Information Fusion for Sequential Recommendation}

\author{Yueqi Xie}
\authornote{Both authors contributed equally.}
\affiliation{%
  \institution{HKUST}
  \country{ }
}
\email{yxieay@connect.ust.hk}

\author{Peilin Zhou}
\authornotemark[1]
\affiliation{%
  \institution{Upstage}
  \country{ }
}
\email{zhoupalin@gmail.com}

\author{Sunghun Kim}
\affiliation{%
  \institution{HKUST}
  \country{ }
}
\email{hunkim@ust.hk}

\begin{abstract}
Side information fusion for sequential recommendation (SR) aims to effectively leverage various side information to enhance the performance of next-item prediction. Most state-of-the-art methods build on self-attention networks and focus on exploring various solutions to integrate the item embedding and side information embeddings before the attention layer.
However, our analysis shows that the early integration of various types of embeddings limits the expressiveness of attention matrices due to a \textit{rank bottleneck} and constrains the flexibility of gradients. Also, it involves mixed correlations among the different heterogeneous information resources, which brings extra disturbance to attention calculation.
Motivated by this, we propose \textbf{D}ecoupled Side \textbf{I}nformation \textbf{F}usion for \textbf{S}equential \textbf{R}ecommendation (DIF-SR), which moves the side information from the input to the attention layer and decouples the attention calculation of various side information and item representation. We theoretically and empirically show that the proposed solution allows higher-rank attention matrices and flexible gradients to enhance the modeling capacity of side information fusion.
Also, auxiliary attribute predictors are proposed to further activate the beneficial interaction between side information and item representation learning.
Extensive experiments on four real-world
datasets demonstrate that our proposed solution stably outperforms state-of-the-art SR models. Further studies show that our proposed solution can be readily incorporated into current attention-based SR models and significantly boost performance.
Our source code is available
at https://github.com/AIM-SE/DIF-SR.

\end{abstract}

\ccsdesc[500]{Information systems~Recommender systems}

\keywords{Sequential Recommendation, Attention Mechanism, Side Information Fusion}


\maketitle

\section{Introduction}
\label{intro}
Sequential recommendation (SR) aims to model users' dynamic preferences from their historical behaviors and make next item recommendations. SR has become an increasingly appealing research topic with wide and practical usage in online scenarios. Multiple deep learning-based solutions~\cite{liu2016context,yu2016dynamic,hidasi2015session} are proposed, and self-attention~\cite{vaswani2017attention} based methods~\cite{kang2018self, sun2019bert4rec,wu2020sse} become mainstream solutions with competitive performances. Among the recent improvements on self-attention based methods, one important branch is related to side information fusion~\cite{zhang2019feature,zhou2020s3,yuan2021icai,liu2021non}. Instead of using item IDs as only item attribute as prior solutions, the side information, such as other item attributes and ratings, is taken into consideration. Intuitively, the highly-related information can benefit recommendations. 
However, how to effectively fuse side information into the recommendation process remains a challenging open issue.
\setlength{\belowcaptionskip}{-0.1cm}   
\begin{figure}[t]
    \centering
    \includegraphics[width=8 cm]{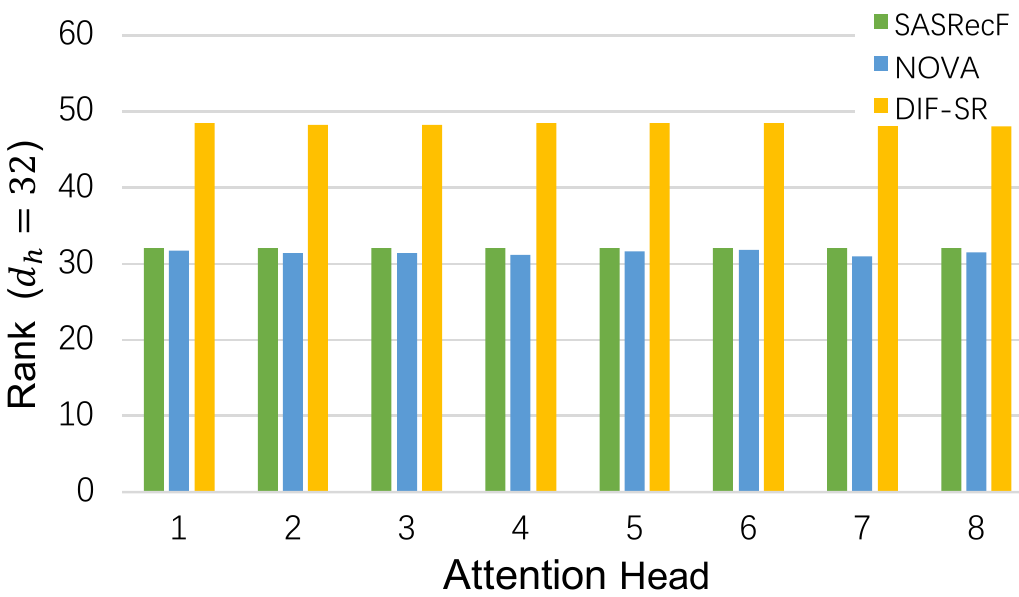}
    \caption{{Rank of attention matrices:} 
    A comparison of the average rank of attention score matrices of early-integrated embedding based solutions, i.e., SASRecF and NOVA, and our proposed DIF-SR.
    The early-integration of embeddings leads to lower rank of the attention matrices and limits the expressiveness.}
    \label{fig:rank}
\end{figure}

Many research efforts have been devoted to fusing side information in various stages of recommendation. Specifically, an early trial FDSA~\cite{zhang2019feature} combines two separate branches of self-attention blocks for item and feature and fuses them in the final stage. S$^3$-Rec~\cite{zhou2020s3} uses self-supervised attribute prediction tasks in the pretraining stage. However, the independent learning of item and side information representation in FDSA and the pretraining strategy in S$^3$-Rec are hard to allow side information to directly interact with item self-attention.

Recently, several studies design solutions integrating side information embedding into the item representation before the attention layer to enable side information aware attention.
ICAI-SR~\cite{yuan2021icai} utilizes attribute-to-item aggregation layer before attention layer to integrate side information into item representation with separate attribute sequential models for training. 
NOVA~\cite{liu2021non} proposes to feed both the pure item id representation and side information integrated representation to the attention layer, where the latter is only used to calculate the attention key and query and keeps the value non-invasive.

Despite the remarkable improvement, several drawbacks remain for the current early-integration based solutions~\cite{liu2021non,yuan2021icai}.
First, we observe that integrating embedding before the attention layer suffers from a \textit{rank bottleneck} of attention matrices, which leads to inferior attention score representation capacity.
It is because of the fact that the rank of attention matrices of prior solutions is inherently bounded by multi-head query-key down-projection size $d_h$, which is usually smaller than the matrices could reach.  
We further theoretically explain such phenomenon in Sec~\ref{method:DIF-theory}.
Second, the attention performed on compound embedding space may lead to random disturbance, where the mixed embeddings from various information resources inevitably attend to unrelated information. The similar drawback for positional encoding in input layer is discussed~\cite{DBLP:conf/iclr/KeHL21,DBLP:conf/sigir/FanLLZXW21}.
Third, since the integrated embedding remains impartible in the whole attention block, early-integrating forces the model to develop complex and heavy integration solutions and training schemes to enable flexible gradients for various side information. With a simple fusion solution, such as widely-used addition fusion, all embeddings share the same gradient for training, which limits the model from learning the relative importance of side-information encodings with respect to item embeddings.

To overcome the limitations, we propose \textbf{D}ecoupled Side \textbf{I}nformation \textbf{F}usion for \textbf{S}equential \textbf{R}ecommendation (DIF-SR). Inspired by the success of decoupled positional embedding~\cite{chen2021simple,DBLP:conf/sigir/FanLLZXW21}, 
we propose to thoroughly explore and analyze the effect of decoupled embedding for side information fusion in the sequential recommendation. 
Specifically, instead of early integration, we move the fusion process from the input to the attention layer. We decouple various side information as well as item embedding through generating key and query separately for every attribute and item in the attention layer. Afterward, we fuse all the attention matrices with fusion function. This simple and effective strategy directly enables our solution to break the \textit{rank bottleneck}, thus enhancing the modeling capacity of the attention mechanism. Fig.~\ref{fig:rank} shows an illustration of rank comparison between current early-integration based solutions and ours with the same embedding size $d$  and head projection size $d_h$. 
Our solution avoids unnecessary randomness of attention caused by the mixed correlation of heterogeneous embeddings.
Also, it enables flexible gradients to adaptively learn various side information in different scenarios.
We further propose to utilize the light \textbf{A}uxiliary \textbf{A}ttribute \textbf{P}redictors (AAP) in a multi-task training scheme to better activate side information to cast beneficial influence on learning final representation.

Experimental results show that our proposed method outperforms the existing basic SR methods~\cite{sun2019bert4rec,hidasi2015session,tang2018personalized,kang2018self} and competitive side information integrated SR methods~\cite{zhou2020s3,liu2021non,yuan2021icai} on four widely-used datasets for sequential recommendation, including Amazon Beauty, Sports, Toys, and Yelp. Moreover, our proposed solution can be incorporated into self-attention based basic SR models easily.
Further study on two representative models~\cite{kang2018self,sun2019bert4rec} shows that significant improvement is achieved when basic SR models are incorporated with our modules. Visualization on attention matrices also offers an interpretation of the rationality of decoupled attention calculation and attention matrices fusion.

Our contribution can be summarized as follows:
\begin{itemize}
    \item We present the DIF-SR framework, which can effectively leverage various side information for sequential recommendation tasks with higher attention representation capacity and flexibility to learn the relative importance of side information.
    \item We propose the novel DIF attention mechanism and AAP-based training scheme, which can be easily incorporated into attention-based recommender systems and boost performance.
    \item We theoretically and empirically analyze the effectiveness of the proposed solution. We achieve state-of-the-art performance on multiple real-world datasets. Comprehensive ablation studies and in-depth analysis demonstrate the robustness and interpretability of our method.

\end{itemize}
\section{Related work}
\subsection{Sequential Recommendation}
Sequential recommendation (SR) models aim to capture users’ preferences from their historical sequential interaction data and make the next-item prediction.
Early SR studies~\cite{DBLP:conf/uai/ZimdarsCM01,rendle2010factorization,kabbur2013fism,he2016fusing} are often based on the Markov Chain assumption and Matrix Factorization methods, which are hard to deal with complex sequence patterns.
Afterward, inspired by the successful application of deep learning techniques on sequential data, many researchers propose to utilize neural networks such as Convolutional Neural Networks (CNNs)~\cite{tang2018personalized,yuan2019simple}, Recurrent Neural Networks (RNNs)~\cite{quadrana2017personalizing, ma2019hierarchical, yan2019cosrec, zheng2019gated, peng2021ham,hidasi2015session}, Graph Neural Networks (GNNs)~\cite{chang2021sequential} and transformer~\cite{kang2018self,sun2019bert4rec,wu2020sse} to model user-item interactions. 
Specifically, self-attention based methods, such as SASRec~\cite{kang2018self} and BERT4Rec~\cite{sun2019bert4rec} have been regarded with strong potential due to their ability to capture long-range dependencies between items. Recently, many improvements on self-attention based solutions are proposed, taking into consideration personalization~\cite{wu2020sse}, item similarity~\cite{DBLP:conf/recsys/0008P020}, consistency~\cite{he2020consistency}, multiple interests~\cite{,DBLP:conf/sigir/FanLLZXW21}, information dissemination~\cite{wu2019dual}, pseudo-prior items augmentation~\cite{liu2021augmenting}, motifs~\cite{cui2021motif}, etc.
However, most of the current SR methods often assume that only the item IDs are available and do not take the side information like item attributes into consideration, ignoring the fact that such highly-related information could provide extra supervision signals. Different from basic SR solutions, our work is tightly related to the side information aware SR, which aims to design effective fusion methods to better utilize various side information.

\subsection{Side Information Fusion for Sequential Recommendation}
Side information aware SR has already been recognized as an important research branch of recommender systems. 
In recent years, side information fusion has also been widely explored in attention-based SR models.
FDSA~\cite{zhang2019feature} adopts different self-attention blocks to encode item and side information, and their representations are not fused until in the final stage.
S$^3$-Rec~\cite{zhou2020s3} adopts pretraining to include the side information. Specifically, two pretraining tasks are designed to utilize these meaningful supervision signals. These approaches verify that side information can beneficially help the next-item prediction. However, they do not effectively and directly use side information to help the attentive aggregation of item representation and prediction.

Several recent works seek to integrate side information embedding to item embedding before the attention layer so that the attentive learning process of final representation can take into consideration side information.
Early works like p-RNN~\cite{hidasi2016parallel} often use simple concatenation to directly inject the side information into item representation. ICAI-SR~\cite{yuan2021icai} utilizes an item-attribute aggregation model to compute the item and attribute embeddings based on the constructed heterogeneous graph and then feeds these fused embeddings to the sequential model to predict the next item. Parallel entity sequential models are used for training. 
NOVA~\cite{liu2021non} proposes to fuse side information without hurting the consistency of item embedding space, where integrated embeddings are only used for keys and queries and values are obtained from pure item id embeddings.
Despite the improvement, we argue that these methods, which build on the integrated embedding and coupled attention calculation among heterogeneous information resources, suffer from the drawbacks including limited representation capacity, inflexible gradient and compound embedding space.


\begin{figure*}[t]
  \centering
  \includegraphics[width=\textwidth]{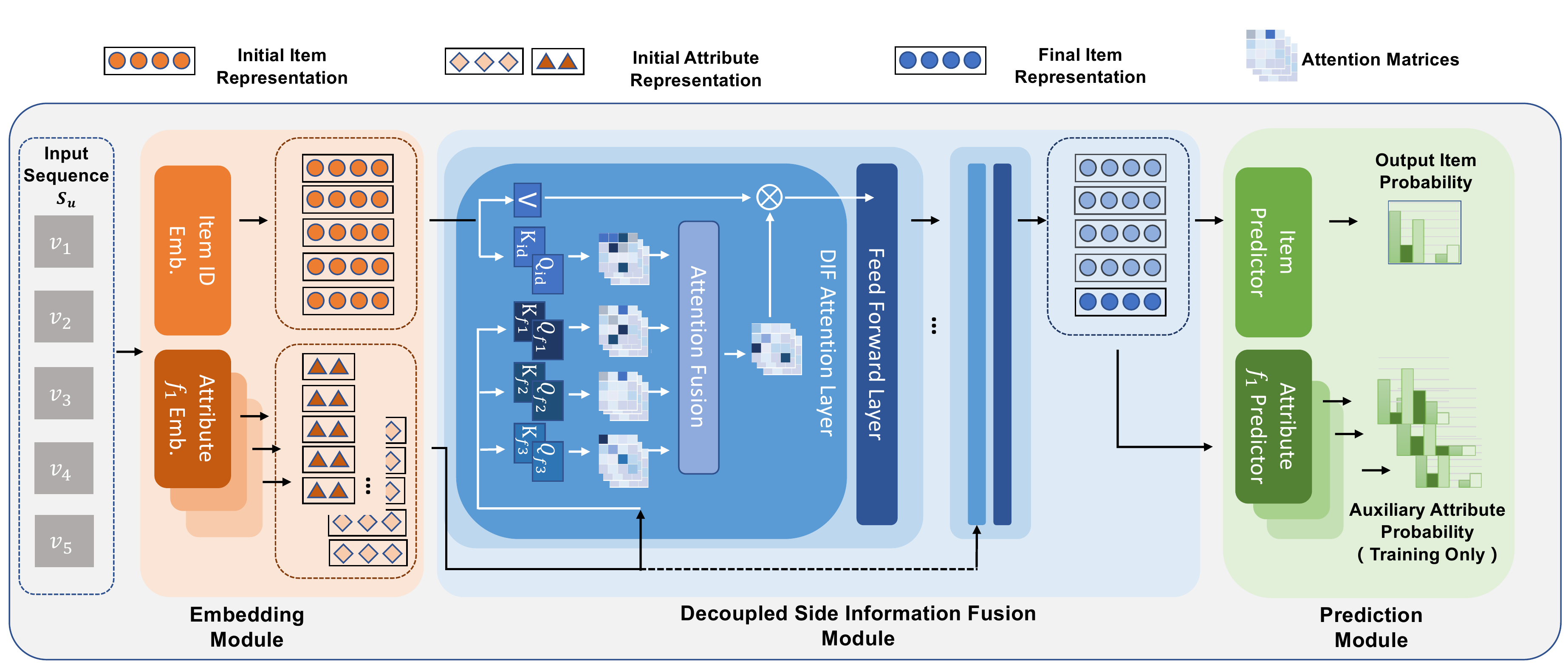}
  \caption{Overview of the proposed network.}
\label{fig:overview}
\end{figure*}

\section{Problem Formulation}
\label{problem}
In this section, we specify the research problem of side information integrated sequential recommendation.
Let $\mathcal{I}$, $\mathcal{U}$ 
denote the sets of items and users respectively. 
For a user $u \in \mathcal{U}$, the user's historical interactions can be represented as:
$\mathbb{S}_u=[v_1, v_2, \dots, v_n]$,
where the term $v_i$ represents the $i$th interaction in the chronologically ordered sequence. 
Side information can be the attributes of users, items, and actions that provide extra information for prediction. 

Following the definition in the prior work~\cite{liu2021non}, side information include item-related information (e.g., brand, category) and behavior-related information (e.g., position, rating). 
Suppose we have $p$ types of side information.
Then, for the side information integrated sequential recommendation, each interaction can be represented as:
$
    v_i=(I_i, f_i^{(1)}, \dots, f_i^{(p)}),
$
where $f_i^{(j)}$ represents $j$th type of the side information of the $i$th interaction in the sequence, and  $I_i$ represents the item ID of the $i$th interaction.
Given this sequence $\mathbb{S}_u$, our goal is to predict the item $I_{pred} \in \mathcal{I}$ the user $u$ will interact with highest possibility:
$I_{pred}=I^{(\hat{k})}$, where
$\hat{k} = {arg\,max}_{k} P(v_{n+1}=(I^{(k)},\cdot) | \mathbb{S}_u).
$

\section{Methodology}
In this section, we present our DIF-SR to effectively and flexibly fuse side information to help next-item prediction.
The overall architecture of DIF-SR is illustrated in Fig.~\ref{fig:overview}, consisting of three main
modules: \textbf{Embedding Module} (Sec.~\ref{method:emb}), \textbf{Decoupled Side Information Fusion Module} (Sec.~\ref{method:DIF}), and \textbf{Prediction Module with AAP} (Sec.~\ref{method:predictor}).


\begin{figure*}[h]
    \centering
    \subfigure[SASRec$_\text{F}$.]{
    \label{fig:sasrecf}
    \includegraphics[width=0.32\linewidth]{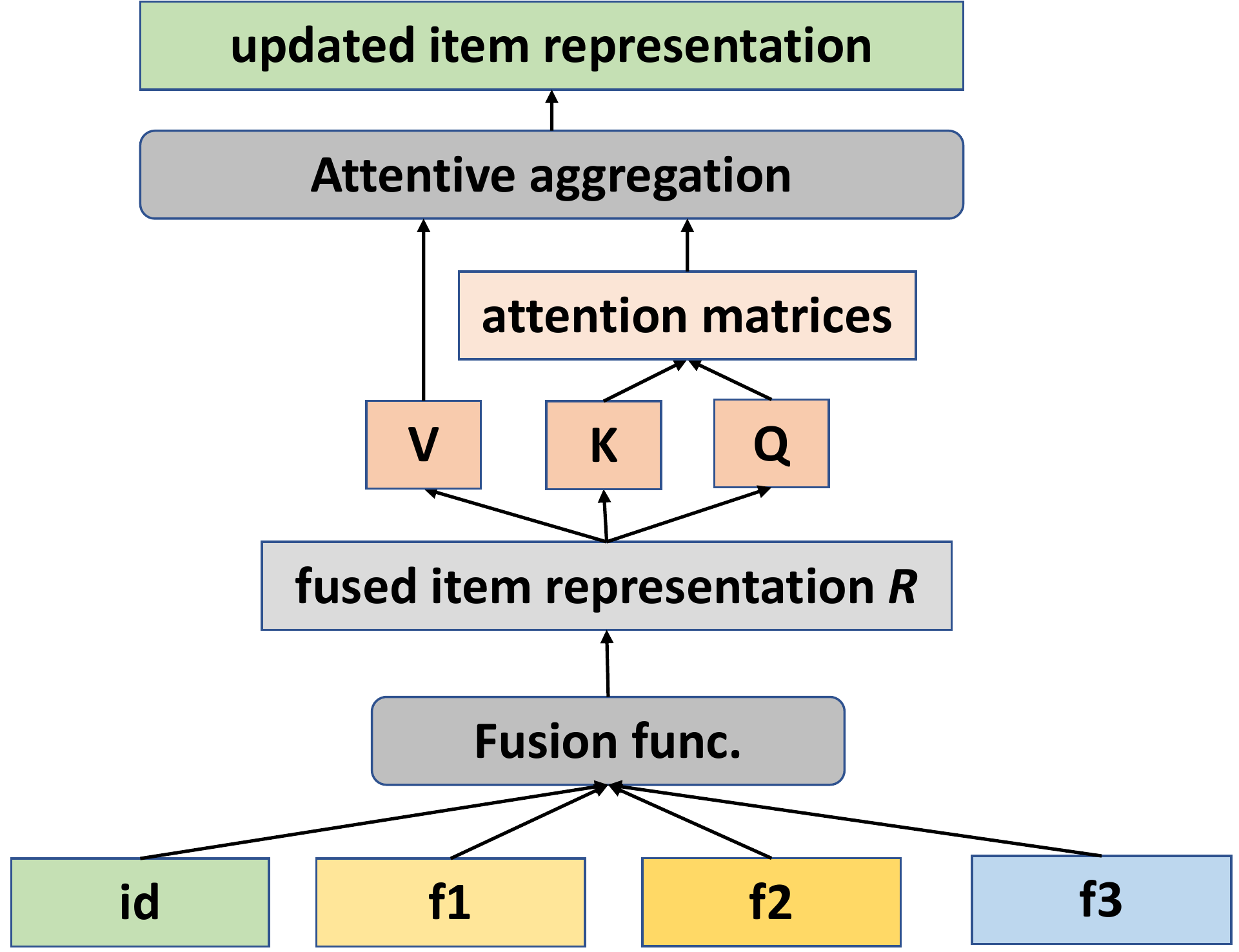}
    }
    \subfigure[NOVA-SR.]{
    \label{fig:nova}
    \includegraphics[width=0.32\linewidth]{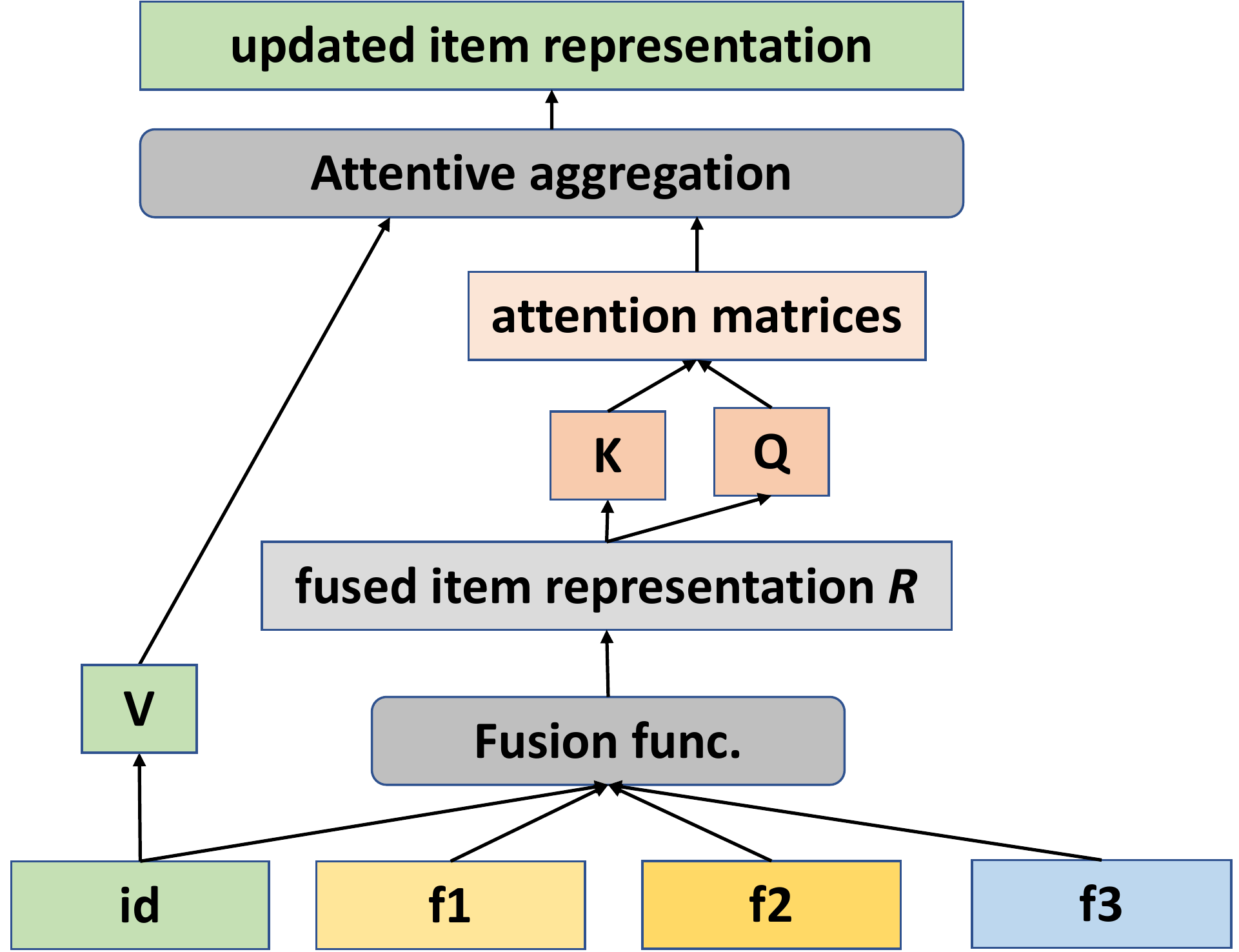}
    }
    \subfigure[DIF-SR.]{
    \label{fig:dif}
    \includegraphics[width=0.32\linewidth]{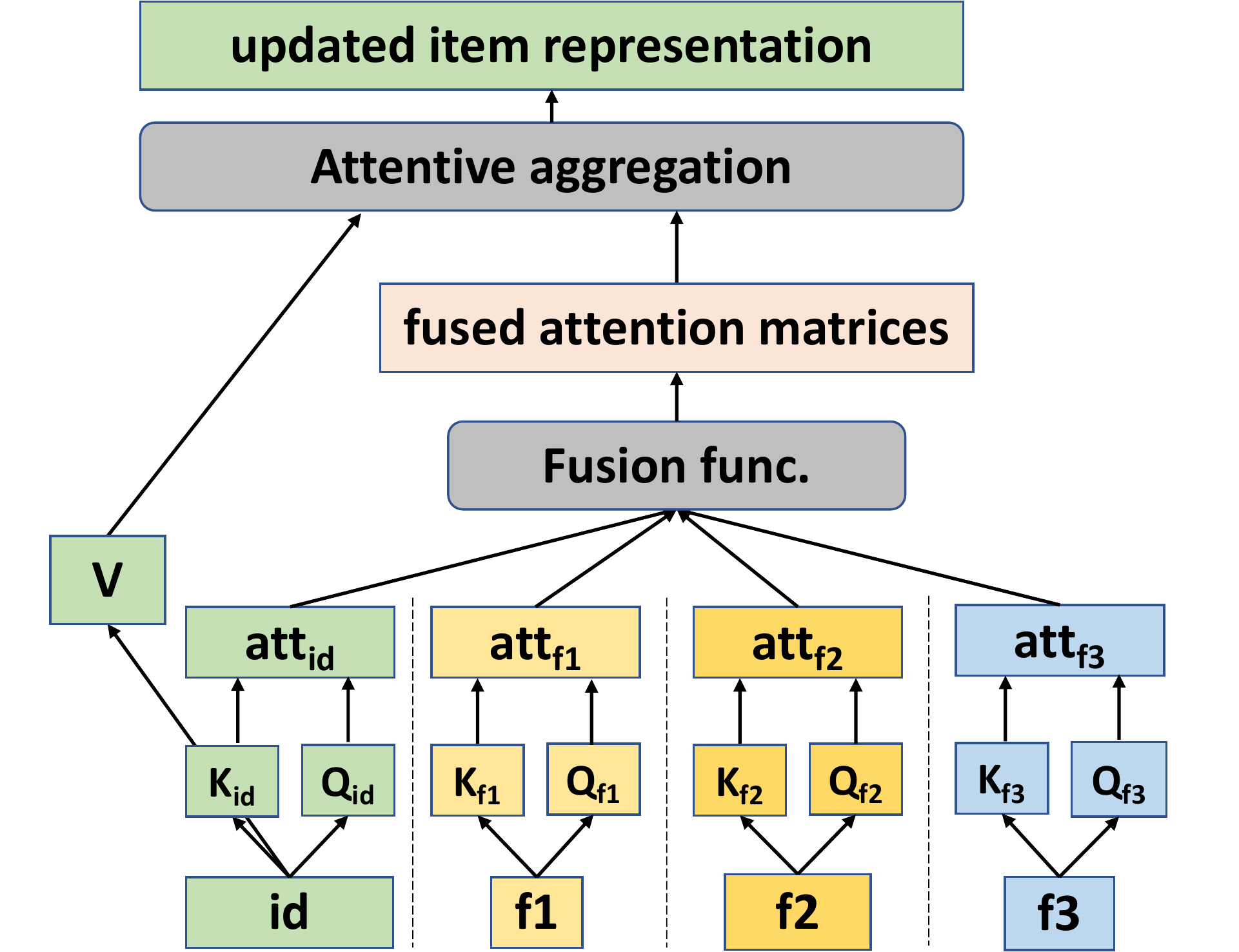}
    }
    \caption{The comparison of item representation learning process of various solutions. (a)  SASRec$_\text{F}$: SASRec$_\text{F}$ fuses side information into item representation and uses the fused item representation to calculate key, query and value.  (b) NOVA-SR: NOVA-SR uses the fused item representation for the calculation of key and query, and keeps value non-invasive. (c) DIF-SR: Instead of early fusion to get fused item representation, the proposed DIF-SR decouples the attention calculation process of various side information to generate fused attention matrices for higher representation power, avoidance of mixed correlation, and flexible training gradient.}
    \label{fig:comp}
\end{figure*}

\subsection{Embedding Module}
\label{method:emb}
In the embedding module, the input sequence $\mathbb{S}_u=[v_1, v_2, \dots, v_n]$ is fed into the item embedding layer and various attribute embedding layers to get the item embedding $E^{ID}$ and side information embeddings $E^{f1}, \dots, E^{fp}$:
\begin{equation}
\begin{aligned}
    E^{ID}=&\mathcal{E}_{id}([I_1,I_2,\dots, I_n]),
\\
    E^{f1}=& \mathcal{E}_{f1}([f_1^{(1)},f_2^{(1)}, \dots,f_n^{(1)}]),
\\
    &\cdots \\
    E^{fp}=& \mathcal{E}_{fp}([f_1^{(p)},f_2^{(p)}, \dots,f_n^{(p)}]),
\end{aligned}
\end{equation}
where $\mathcal{E_\cdot}$ represents the corresponding embedding layer that encodes the item and different item attributes into vectors. The look-up embedding matrices can be represented as $M_{id}\in\mathbb{R}^{|\mathcal{I}|\times d}, M_{f1}\in\mathbb{R}^{|f1|\times d_{f1}} ,\dots,M_{fp}\in\mathbb{R}^{|fp|\times d_{fp}}$, and $|\cdot|$ represents the corresponding total number of different items and various side information, and $d$ and $d_{f1},\dots,d_{fp}$ denote the dimension of embedding of item and various side information. Note that supported by the operations in proposed DIF module, the embedding dimensions are flexible for different types of attributes. It is further verified in Sec.~\ref{exp:hyper} that we can apply much smaller dimension for attribute than item to greatly improve efficiency of the network without harming the performance. Then the embedding module gets the output embeddings $E^{ID}\in\mathbb{R}^{n\times d}, E^{f1}\in\mathbb{R}^{n\times d_{f1}}, \dots, E^{fp}\in\mathbb{R}^{n\times d_{fp}}$.

\subsection{Decoupled Side Information Fusion Module}
\label{method:DIF}
We first specify the overall layer structure of the module in Sec.~\ref{method:DIF-layer}.
To better illustrate our proposed DIF attention, we discuss self-attentive learning process of the prior solutions~\cite{liu2021non,kang2018self} in Sec.~\ref{method:DIF-prior}.
Afterward, we comprehensively introduce the proposed DIF attention in Sec.~\ref{method:DIF-ours}.
Finally, we provide theoretical analysis on the enhancement of DIF on expressiveness of the model regarding the rank of attention matrices and flexibility of gradients in Sec.~\ref{method:DIF-theory}.

\subsubsection{Layer Structure.}
\label{method:DIF-layer}
As shown in Fig.~\ref{fig:overview}, the Decoupled Side Information Fusion Module contains several stacked blocks of sequential combined DIF Attention Layer and Feed Forward Layer. The block structure is the same as SASRec~\cite{kang2018self}, except that we replace the original multi-head self-attention with the multi-head DIF attention mechanism.
Each DIF block involves two types of input, namely current item representation and auxiliary side information embeddings, and then output updated item representation. Note that the auxiliary side information embeddings are not updated per layer to save computation while avoiding overfitting. Let $R_{i}^{(ID)}\in\mathbb{R}^{n\times d} $ denote the input item representation of block $i$. 
The process can be formulated as:
\begin{equation}
\begin{aligned}
   R_{i+1}^{(ID)} =\textrm{LN}(\textrm{FFN}(\textrm{DIF}(R_{i}^{(ID)},E^{f1},\dots, E^{fp}))),
\end{aligned}
\end{equation}
\begin{equation}
\begin{aligned}
R_1^{(ID)}=&E^{ID},
\end{aligned}
\end{equation}
where $\textrm{FFN}$ represents a fully connected feed-forward network and $\textrm{LN}$ denotes layer normalization.

\subsubsection{Prior Attention Solutions.}
\label{method:DIF-prior}
Fig.~\ref{fig:comp} shows the comparison of the prior solutions to fuse side information to the updating process of item representation. Here we focus on the self-attention calculation, which is the main difference of several solutions. 

\textbf{SASRec$_\text{F}$}: 
As shown in Fig.~\ref{fig:sasrecf}, the solution directly fuses side information embedding to the item representation and performs vanilla self-attention upon the integrated embedding, which is extended from original SASRec~\cite{kang2018self}. 

With the input length $n$, hidden size $d$, multi-head query-key down-projection size $d_h$, we can let $R\in\mathbb{R}^{n\times d}$ denote the integrated embedding, and $W_Q^i,W_K^i,W_V^i\in\mathbb{R}^{d\times d_h}, i \in [h]$ denote the query, key, and value projection matrices for $h$ heads ($d_h=d/h$), then the calculation of attention score can be formalized as:

\begin{equation}
\begin{aligned}
\text{SAS\_att}^i = (RW_Q^i)(RW_K^i)^\top.
\end{aligned}
\end{equation}
Then the output for each head can be represented as:
\begin{equation}
\begin{aligned}
\text{SAS\_head}^i = \sigma(\frac{\text{SAS\_att}^i}{\sqrt{d}})(RW_V^i),
\end{aligned}
\end{equation}
where $\sigma$ denotes Softmax function.

Despite that this solution allows side information to directly influence the item representation learning process, it is observed that such method has the drawback of invasion of item representation~\cite{liu2021non}.

\textbf{NOVA}: To solve the mentioned problem, the work~\cite{liu2021non} proposes to utilize the non-invasive fusion of side information. As shown in Fig.~\ref{fig:nova}, NOVA calculates $Q$ and $K$ from the integrated embeddings $R\in\mathbb{R}^{n\times d}$, and $V$ from the pure item ID embeddings $R^{(ID)}\in\mathbb{R}^{n\times d}$. Similarly, $W_Q^i,W_K^i,W_V^i\in\mathbb{R}^{d\times d_h}, i \in [h]$ denote the query,key, and value projection matrices:
\begin{equation}
\begin{aligned}
\text{NOVA\_att}^i = (RW_Q^i)(RW_K^i)^\top
\end{aligned}
\end{equation}
Then the output for each head can be represented as:
\begin{equation}
\begin{aligned}
\text{NOVA\_head}^i = \sigma(\frac{\text{NOVA\_att}^i}{\sqrt{d}})(R^{(ID)}W_V^i)
\end{aligned}
\end{equation}

\subsubsection{DIF Attention.}
\label{method:DIF-ours}
We argue that although NOVA solves the invasion problem of value, using integrated embedding for the calculation of key and value still suffers from compound attention space as well as the degradation of expressiveness regarding the rank of attention matrix and training gradient flexibility. The supportive theoretical analysis is shown in Sec.~\ref{method:DIF-theory}.

Therefore, in contrast to the prior studies that inject attribute embedding into item representation for a mixed representation, we propose to leverage decoupled side information fusion solution. As shown in Fig.~\ref{fig:dif}, in our proposed solution, all attributes attend upon themselves to generate the decoupled attention matrices, which are then fused to the final attention matrices. The decoupled attention calculation improves the model's expressiveness through breaking the \textit{rank bottleneck} of attention matrices bounded by head projection size $d_h$. It also avoids inflexible gradients and uncertain cross relationships among different attributes and the item to enable a reasonable and stable self-attention.


Given the input length $n$, item hidden size $d$, multi-head query-key down-projection size $d_h$, we have $W_Q^i,W_K^i,W_V^i\in\mathbb{R}^{d\times d_h}, i \in [h]$ denote the query, key, and value projection matrices for $h$ heads ($d_h=d/h$) for item representation $R^{(ID)}\in\mathbb{R}^{n\times d}$. Then attention score for item representation is then calculated with:
\begin{equation}
\begin{aligned}
    \text{att}_{ID}^i=(R^{(ID)}W_Q^i)(R^{(ID)}W_K^i)^\top.
\end{aligned}
\end{equation}

Different from prior works, we also generate multi-head attention matrices for every attribute with attribute embeddings $E^{f1}\in\mathbb{R}^{n\times d_{f1}},\dots,E^{fp}\in\mathbb{R}^{n\times d_{fp}}$. Note that we have $d_{fj} \le d, j \in [p]$ to avoid over-parameterization and reduce computation overhead. Then we have corresponding
$W_Q^{(fj)i},W_K^{(fj)i},W_V^{(fj)i}\in\mathbb{R}^{d\times d_{hj}}, i \in [h], j \in[p]$ denote the query,key, and value projection matrices for $h$ heads ($d_{hj}=d_{fj}/h$):
\begin{equation}
\begin{aligned}
    \text{att}_{f1}^i =&(E^{f1}W_Q^{(f1)i})(E^{f1}W_K^{(f1)i})^\top,  \\
    &\dots, \\
    \text{att}_{fp}^i =&(E^{fp}W_Q^{(fp)i})(E^{fp}W_K^{(fp)i})^\top.
\end{aligned}
\end{equation}




Then our DIF attention fuses all the attention matrices with the fusion function $\mathcal{F}$ explored in the prior work~\cite{liu2021non}, including addition, concatenation, and gating, and gets the output for each head as:

\begin{equation}
\label{equ:fusion}
\begin{aligned}
\textrm{DIF\_att}^i=\mathcal{F}(\text{att}_{ID}^i, \text{att}_{f1}^i ,\dots, \text{att}_{fp}^i),\\
\text{DIF\_head}^i = \sigma(\frac{\text{DIF\_att}^i}{\sqrt{d}})(R^{(ID)}W_V^i).
\end{aligned}
\end{equation}

Finally, the outputs of all attention heads are concatenated and fed to the feed-forward layer.      

\subsubsection{Theoretical Analysis}
\label{method:DIF-theory}
In this section, we extend the analysis for positional embedding in ~\cite{chen2021simple} to theoretically analyze the proposed DIF and early fusion solution in the previous models~\cite{kang2018self,liu2021non}. See Appendix~\ref{appendix:proofs}  for the proof.

We first discuss the \textit{expressiveness} of the model regarding the rank of attention matrices for DIF and the prior solutions.

\begin{theorem} \label{thm:rank}
Let $W_Q, W_K \in\mathbb{R}^{d\times d_h}$ denote the projection matrices with head projection size $d_h$, $W_Q^{fj}, W_K^{fj}$ denote projection matrices for attribute $j$ with head projection size $d_{hj}$, and $d_{hj} \le d_h$, $d$ and $n \ge d_h + \Sigma_{j=1}^{p} d_{hj}$.
Let $\text{att} = R W_Q W_K^\top R^\top$ be the attention matrices for integrated representation $R\in\mathbb{R}^{n\times d}$. 
Let $DIF\_{att} = R^{(ID)} W_Q W_K^\top R^{(ID)\top} + \Sigma_{j=1}^{p}(E^{fj}W_Q^{fj}W_K^{fj\top}E^{fj\top})$ be the attention matrices for decoupled representation $R^{(ID)}\in\mathbb{R}^{n\times d}$, $E^{f1}\in\mathbb{R}^{n\times d_{f1}},\dots,E^{fp}\in\mathbb{R}^{n\times d_{fp}}$. Then, for any $R, W_Q, W_K$ $$rank(att) \leq d_h.$$
There exists a choice of parameters such that
 $$rank(DIF\_{att}) = d_h + \Sigma_{j=1}^{p}d_{hj} > d_h.$$

\end{theorem}
\paragraph{\textbf{Remarks.}} This theorem shows that early fusion of side information limits the rank of attention matrices by $d_h$, which is usually a smaller value than the attention matrices could reach in multi-head self-attention. And our solution breaks such rank bottleneck through fusing the decoupled attention scores for items and attributes. Fig.~\ref{fig:rank} also offers an experimental illustration of rank comparison. Higher-rank attention matrices inherently 
enhance DIF-SR with stronger model expressiveness.

Then, we discuss the \textit{training flexibility} of integrated embedding based solution. We claim that with simple addition fusion solution, both SASRec$_\text{F}$ and NOVA limit the flexibility of gradients. Let $E^{ID}, E^{f1},\dots, E^{fp}\in\mathbb{R}^{n\times d}$ denote the input item and attribute embeddings. 
For SASRec$_\text{F}$, all the embeddings are added up and fed into the model $G$. For given input and label $y$, the loss function can be denoted as:
\begin{equation}
\label{eqsas}
L = \ell(G(E^{ID} + \Sigma_{i=1}^{p}E^{fi}), y).
\end{equation}
For NOVA, the attribute embeddings are first added up as inseparable and then fused to item representation in every layer, thus for model $G$ and label $y$, the loss function can be denoted as:
\begin{equation}
\label{eqnova}
L = \ell(G(E^{ID}, \Sigma_{i=1}^{p}E^{fi}), y).
\end{equation}

\begin{theorem}\label{thm:gradients}
Let $E^{ID}, E^{f1},\dots, E^{fp}\in\mathbb{R}^{n\times d}$ denote the input item and attribute embeddings. For the loss function in Equ.\eqref{eqsas}, the gradients for $E^{ID}, E^{f1},\dots, E^{fp}$  remain same for any function $\ell$ and model $G$, and any input and label $y$.  For the loss function in Equ.\eqref{eqnova}, the gradients for $E^{f1},\dots, E^{fp}$  remain same for any function $\ell$ and model $G$, and any input and label $y$.
\end{theorem}

\noindent \paragraph{\textbf{Remarks.}} This theorem shows that with simple addition fusion, the gradients are the same for the input item embeddings and attribute embeddings for SASRec, and all types of attribute embeddings share the same gradient for NOVA. This suggests that to enable flexible gradients, more complex and heavy fusion solutions need to be involved for early-integration based methods compared with ours.



\subsection{Prediction Module with AAP}
\label{method:predictor}
With the final representation $R_L^{(ID)}$ that encodes the sequence information with the help of side information, we take the last element of $R_L^{(ID)}$, denoted as $R_L^{(ID)}[n]$, to
estimate the probability of user $u$ to interact with each item in the item vocabulary.
The item prediction layer can be formulated as:
\begin{equation}
    \hat{y} = \text{softmax}(M_{id} R_L^{(ID)}[n]^\top),
\end{equation}
where $\hat{y}$ represents the $|\mathcal{I}|$-dimensional probability, and  $M_{id}\in\mathbb{R}^{|\mathcal{I}|\times d}$ is the item embedding table in the embedding layer.

During training, we propose to use Auxiliary Attribute Predictors (\textbf{AAP}) for attributes (except position information) to further activate the interaction between auxiliary side information and item representation.
Note that different from prior solutions that make predictions with separate attribute embedding~\cite{yuan2021icai} or only use attributes for pretraining~\cite{zhou2020s3},  we propose to apply multiple predictors directly on final representation to force item representation to involve useful side information.
As verified in Sec.~\ref{exp:components}, AAP can help further improve the performance, especially when combining with DIF. We attribute this to the fact that AAP is designed to enhance the attributes' informative influence on self-attentive item representation learning, while early-integration based solutions do not support such influence.

The prediction for attribute $j$ can be represented as:
\begin{equation}
\begin{aligned}
    \hat{y}^{(fj)} =& \sigma(W_{fj} R_L^{(ID)}[n]^\top + b_{fj}),
\end{aligned}
\end{equation}
where $\hat{y}^{(fj)}$ represents the $|fj|$-dimensional probability, $W_{fj}\in\mathbb{R}^{|fj|\times d_{fj}}$ and $b_{fj}\in\mathbb{R}^{|fj|\times 1}$ are learnable parameters,
and $\sigma$ is the sigmoid function.
Here we compute the item loss $L_{id}$ using cross entropy to measure the difference between prediction $\hat y$ and ground truth $y$:
\begin{equation}
\begin{aligned}
	L_{id}=&-\sum_{i=1}^{|\mathcal{I}|}{y_i \log (\hat y_i)}
\end{aligned}
\end{equation}
And following \cite{adsr}, we compute the side information loss $L_{fj}$ for $j$th type using binary cross entropy to support multi-labeled attributes:
\begin{equation}
\begin{aligned}
	L_{fj}=&-\sum_{i=1}^{|fj|}{y^{(fj)}_i \log (\hat y^{fj}_i)+(1-y^{(fj)}_i)\log (1-\hat y^{(fj)}_i)}, \\
\end{aligned}
\end{equation}
Then, the combined loss function with balance parameter $\lambda$ can be formulated as:
\begin{equation}
\label{eq:all_loss}
\begin{aligned}
	L =& L_{id} + \lambda \Sigma_{j=1}^{p}L_{fj}.
\end{aligned}
\end{equation}

\section{Experiments}
We conduct extensive experiments on four real-world and widely-used datasets to answer the following research questions.
\begin{itemize}
    \item \textbf{RQ1:} Does DIF-SR outperform the current state-of-the-art basic SR methods and side information integrated SR methods?
    \item \textbf{RQ2:} Can the proposed DIF and AAP be readily incorporated into the state-of-the-art self-attention based models and boost the performance?
    \item \textbf{RQ3:} What is the effect of different components and hyperparameters in the DIF-SR framework?
    \item \textbf{RQ4:} Does the visualization of attention matrices fusion of DIF provide evidence for superior performance?
\end{itemize}
\subsection{Experimental Settings}
\subsubsection{Dataset.}
The experiments are conducted on four real-world and widely used datasets. 
\begin{itemize}
    \item \textbf{Amazon Beauty}, \textbf{Sports} and \textbf{Toys}\footnote{http://jmcauley.ucsd.edu/data/amazon/}. These datasets are constructed from Amazon review datasets~\cite{DBLP:conf/sigir/McAuleyTSH15}. 
    Following baseline \cite{yuan2021icai}, we utilize fine-grained categories of the goods and position information as attributes for all these three datasets. 
    \item \textbf{Yelp}~\footnote{https://www.yelp.com/dataset}, which is a well-known business recommendation dataset. Following ~\cite{zhou2020s3}, we only retain the transaction records after Jan. 1st, 2019 in our experiment, and the categories of businesses and position information are regarded as attributes. 
\end{itemize}

Following the same pre-processing steps used in ~\cite{kang2018self,zhou2020s3,yuan2021icai}, we remove all items and users that occur less than five times in these datasets. All the interactions are regarded as implicit feedback. 
The statistics of all these four datasets after preprocessing are summarized in Tab.~\ref{tab:datasets}.
\begin{table}
    \small
	\caption{Statistics of the datasets after preprocessing.}
	\label{tab:datasets}
	\setlength{\tabcolsep}{0.6mm}{
	\begin{tabular}{lrrrr}
	\toprule
	Dataset  & Beauty & Sports & Toys & Yelp  \\
	\midrule
	\# Users & 22,363 & 35,598 & 19,412 & 30,499\\
	\# Items & 12,101 & 18,357 & 11,924 & 20,068 \\
	\# Avg. Actions / User & 8.9 & 8.3 & 8.6 & 10.4\\
	\# Avg. Actions / Item & 16.4 & 16.1 & 14.1 & 15.8\\
	\# Actions & 198,502 & 296,337 & 167,597 & 317,182 \\
	Sparsity & 99.93\% & 99.95\% & 99.93\% & 99.95\%\\
	\bottomrule
	\end{tabular}
	}
\end{table}
\subsubsection{Evaluation Metrics.}
In our experiments, following the prior works~\cite{kang2018self,sun2019bert4rec}, we use \textit{leave-one-out} strategy for evaluation. Specifically, for each user-item interaction sequence, the last two items are reserved as validation and testing data, respectively, and the rest are utilized for training SR models. The performances of SR models are evaluated by top-$K$ Recall (Recall@$K$) and top-$K$ Normalized Discounted Cumulative Gain (NDCG@$K$) with $K$ chosen from $\{10,20\}$, which are two commonly used metrics. As suggested by~\cite{krichene2020sampled,dallmann2021case}, we evaluate model performance in a full ranking manner for a fair comparison. The ranking results are obtained over the whole item set rather than sampled ones.
\subsubsection{Baseline Models.}
\label{baselines}
We choose two types of state-of-the-art methods for comparison, including strong basic SR methods and recent competitive side information integrated SR methods.
The baselines are introduced as follows.
\begin{itemize}
    \item \textbf{GRU4Rec}~\cite{hidasi2015session}: A session-based recommender system that uses RNN to capture sequential patterns. 
    \item \textbf{GRU4Rec$_\text{F}$}: An enhanced version of GRU4Rec that considers the side information to improve the performance.
    \item \textbf{Caser}~\cite{tang2018personalized}: A CNN-based model that uses horizontal and vertical convolutional filters to learn multi-level patterns and user preferences.
    \item \textbf{BERT4Rec}~\cite{sun2019bert4rec}: A bidirectional self-attention network that uses Cloze task to model user behavior sequences.
    \item \textbf{SASRec}~\cite{kang2018self}: An attention-based model that uses self-attention network for the sequential recommendation.
    \item \textbf{SASRec$_\text{F}$}: An extension of SASRec, which fuses the item and attribute representation through concatenation operation before feeding them to the model.
    \item \textbf{S$^3$-Rec}~\cite{zhou2020s3}: A self-supervised learning based model with four elaborately designed optimization objectives to learn the correlation in the raw data.
    \item \textbf{ICAI-SR}~\cite{yuan2021icai}: A generic framework that carefully devised attention-based Item-Attribute Aggregation model (IAA) and Entity Sequential (ES) models for exploiting various relations between items and attributes. For a fair comparison, we instantiate ICAI-SR by taking SASRec as ES models in our experiments.
    \item \textbf{NOVA}~\cite{liu2021non}: A framework that adopts non-invasive self-attention (NOVA) mechanism for better attention distribution learning. Similar to ICAI-SR, we implement NOVA mechanism on SASRec for a fair comparison.
\end{itemize}
\begin{table*}[t]
    \caption{Overall performance. Bold scores represent the highest results of all methods. Underlined scores stand for the highest results from previous methods. }
    \small
    \begin{tabular}{c|l|ccccccccc|c}
    \toprule
         Dataset& Metric  & GRU4Rec & Caser & BERT4Rec & GRU4Rec$_\text{F}$ & SASRec & SASRec$_\text{F}$ & S$^{3}$Rec & NOVA& ICAI & DIF-SR \\
         \midrule
         \multirow{4}*{Beauty}&Recall@10& 0.0530 & 0.0474 & 0.0529  & 0.0587  & 0.0828  & 0.0719 & 0.0868 & \underline{0.0887} & 0.0879& \textbf{0.0908}\\
         &Recall@20& 0.0839 & 0.0731 & 0.0815 & 0.0902  & 0.1197 & 0.1013 & 0.1236 & \underline{0.1237} & 0.1231 & \textbf{0.1284}\\
         &NDCG@10& 0.0266 & 0.0239 & 0.0237  & 0.0290  & 0.0371  & 0.0414 & \underline{0.0439} & \underline{0.0439} & \underline{0.0439} & \textbf{0.0446}\\
         &NDCG@20& 0.0344 & 0.0304 & 0.0309  & 0.0369  & 0.0464  & 0.0488 & \underline{0.0531} & 0.0527& 0.0528& \textbf{0.0541}\\
         \midrule
         \multirow{4}*{Sports}&Recall@10& 0.0312 & 0.0227 & 0.0295  & 0.0394  & 0.0526  & 0.0435 & 0.0517 & \underline{0.0534} & 0.0527& \textbf{0.0556}\\
         &Recall@20& 0.0482 & 0.0364 & 0.0465  & 0.0610  & \underline{0.0773}  & 0.0640 & 0.0758 & 0.0759& 0.0762 & \textbf{0.0800}\\
         &NDCG@10& 0.0157 & 0.0118 & 0.0130  & 0.0199  & 0.0233  & 0.0235 & 0.0249 & \underline{0.0250} & 0.0243& \textbf{0.0264}\\
         &NDCG@20& 0.0200 & 0.0153 & 0.0173  & 0.0253  & 0.0295  & 0.0286 & \underline{0.0310} & 0.0307& 0.0302& \textbf{0.0325}\\
         \midrule
         \multirow{4}*{Toys}&Recall@10& 0.0370 & 0.0361 & 0.0533  & 0.0492  & 0.0831  & 0.0733 & 0.0967 & \underline{0.0978} & 0.0972& \textbf{0.1013}\\
         &Recall@20& 0.0588 & 0.0566 & 0.0787  & 0.0767  & 0.1168  & 0.1052 & \underline{0.1349} & 0.1322& 0.1303& \textbf{0.1382}\\
         &NDCG@10& 0.0184 & 0.0186 & 0.0234  & 0.0246  & 0.0375  & 0.0417 & 0.0475 & \underline{0.0480} & 0.0478& \textbf{0.0504}\\
         &NDCG@20& 0.0239 & 0.0238 & 0.0297  & 0.0316  & 0.0460  & 0.0497 & \underline{0.0571} & 0.0567& 0.0561& \textbf{0.0597}\\
         \midrule
         \multirow{4}*{Yelp}&Recall@10& 0.0361 & 0.0380 & 0.0524  & 0.0361  & 0.0650  & 0.0413 & 0.0589 & \underline{0.0681} & 0.0663& \textbf{0.0698}\\
         &Recall@20& 0.0592 & 0.0608 & 0.0756  & 0.0578  & 0.0928  & 0.0675 & 0.0902 & \underline{0.0964} & 0.0940& \textbf{0.1003}\\
         &NDCG@10& 0.0184 & 0.0197 & 0.0327  & 0.0182  & 0.0401  & 0.0216 & 0.0338 & \underline{0.0412} & 0.0400& \textbf{0.0419}\\
         &NDCG@20& 0.0243 & 0.0255 & 0.0385  & 0.0236  & 0.0471  & 0.0282 & 0.0416 & \underline{0.0483} & 0.0470& \textbf{0.0496}\\
         \bottomrule
    \end{tabular}
    \label{tab:overall}
\end{table*}
\subsubsection{Implementation Details.}
All the baselines and ours are implemented based on the popular recommendation framework RecBole~\cite{zhao2021recbole} and evaluated with the same setting.
For all baselines and our proposed method, we train them with Adam optimizer for 200 epochs, with a batch size of 2048 and a learning rate of 1e-4. 
The hidden size of our DIF-SR and other attention-based baselines are all set to $256$. 
For other hyperparameters, we apply grid-search to find the best config for our model and baseline models that involve the following hyperparameters. The searching space is:  $attribute\_embedding\_size\in \{16,32,64,128,256\}$, $num\_heads\in \{2,4,8\}$, $num\_layers\in \{2,3,4\}$. Balance parameter $\lambda$ in Equ.\eqref{eq:all_loss} is chosen from $\{5,10,15,20,25\}$, and fusion function $\mathcal{F}$ in Equ.\eqref{equ:fusion} for NOVA and ours is chosen from addition, concatenation and gating~\cite{liu2021non}.

\subsection{Overall Performance (RQ1)}
The overall performance of different methods on all datasets is summarized in Table~\ref{tab:overall}. Based on these results, we can observe that:
For four basic SR baselines, it can be observed that SASRec outperforms other methods by a large margin, while BERT4Rec is better than or close to GRU4Rec and Caser in most cases, demonstrating the superiority of attention-based methods on sequential data. 
Note that although BERT4Rec is proposed to be an advanced version of SASRec, its performance is not so good as SASRec under the full ranking evaluation setting, which is also found in prior works \cite{zhou2020s3,li2021lightweight}.
We attribute this phenomenon to the mismatch between masked item prediction and the intrinsic autoregressive nature of SR. And the superior performance of BERT4Rec under the popularity-based sampling strategy in the original paper may be because the bidirectional encoder with Cloze task used by BERT4Rec could learn better representations for popular items for evaluation.

Based on the aforementioned findings, we implement all attention-based side information aware baselines based on the same attentive structure from SASRec for fair comparison under the full ranking setting. The reimplementation details are discussed in Sec.~\ref{baselines}.
For the side information aware baselines, it can be found that the simple early fusion solution for GRU4Rec$_\text{F}$ and SASRec$_\text{F}$ does not always improve the performance comparing with the version without side information, which is in line with our analysis that early-integrated representation forces the model to design complex merging solution, otherwise it will even harm the prediction performance. Also, the recently proposed side information aware SR models, i.e., S$^3$Rec, NOVA, ICAI, get better and comparative performance, suggesting that with carefully-designed fusion strategy, side information can improve prediction performance.

Finally, it is clear to see that our proposed DIF-SR consistently outperforms both state-of-the-art SR models and side information integrated SR models on all four datasets in terms of all evaluation metrics. 
Different from the baselines, we decouple the attention calculation process of side information and propose to add attribute predictors upon the learned item representation. DIF inherently enhances the expressiveness of self-attention with higher-rank attention matrices, avoidance of mixed correlation, and flexible training gradient, while AAP further strengthens the mutual interactions between side information and item information during the training period. 
These results demonstrate the effectiveness of the proposed solution to improve the performance by leveraging side information for the attention-based sequential recommendation.
\begin{table}[t]
    \centering
    \caption{Performance comparison of self-attention based sequential models with their DIF \& AAP incorporated version on Beauty, Sports and Toys datasets.}
    \resizebox{\linewidth}{!}{
    \begin{tabular}{c|c|c c|c c|c c}
    \toprule
    \multirow{2}*{\textbf{Dataset}} & \multirow{2}*{\textbf{Metric}} & \multicolumn{2}{c|}{DIF-SASRec} & \multicolumn{2}{c|}{DIF-BERT4Rec} & \multicolumn{2}{c}{Improv.} \\
              ~ & & w/o & w & w/o & w & DIF-SASRec & DIF-BERT4Rec  \\
    \midrule
    \multirow{6}*{Beauty}
              ~ & Recall@10 &  0.0828 &  0.0908 &  0.0529 &  0.0579 &  9.66\% &  9.45\% \\
              ~ & Recall@20 &  0.1197 &  0.1284 &  0.0815 &  0.0915 &  7.27\% &  12.27\% \\  
              ~ & NDCG@10 &  0.0371 &  0.0446 &  0.0237 &  0.0279 &  20.22\% &  17.72\%\\   
              ~ & NDCG@20 &  0.0464 &  0.0541 &  0.0309 &  0.0363 &  16.59\% &  17.48\%\\
    \midrule
    \multirow{6}*{Sports}
              ~ & Recall@10 &  0.0526 &  0.0556 &  0.0295 &  0.0394 &  5.70\% &  33.56\% \\
              ~ & Recall@20 &  0.0773 &  0.0800 &  0.0465 &  0.0611 &  3.49\% &  31.40\% \\ 
              ~ & NDCG@10 &  0.0233 &  0.0264 &  0.0130 &  0.0198 &  13.30\% &  52.31\% \\   
              ~ & NDCG@20 &  0.0295 &  0.0325 &  0.0173 &  0.0252 &  10.17\% &  45.66\% \\
    \midrule
    \multirow{6}*{Toys}
              ~ & Recall@10 &  0.0831 &  0.1013 &  0.0533 &  0.0599 &  21.90\% &  12.38\% \\
              ~ & Recall@20 &  0.1168 &  0.1382 &  0.0787 &  0.0851 &  18.32\% &  8.13\% \\ 
              ~ & NDCG@10 &  0.0375 &  0.0504 &  0.0234 &  0.0324 &  34.40\% &  38.46\% \\   
              ~ & NDCG@20 & 0.0460 &  0.0597 & 0.0297 & 0.0387 & 29.78\% & 30.30\% \\
    \bottomrule
          \end{tabular}
         }
    \label{tab:improvement}
\end{table}


\subsection{Enhancement Study (RQ2)}
With the simple and effective design, we argue that DIF and AAP can be easily incorporated into any self-attention based SR model and boost the performance. To verify this, we conduct experiments on two representative models: SASRec and BERT4Rec, which represent the unidirectional model and bidirectional model, respectively. 

As shown in Tab.~\ref{tab:improvement}, with our design, the enhanced models significantly outperform the original models. Specifically, DIF-BERT4Rec achieves $\textbf{18.46\%}$ and $\textbf{36.16\%}$ average relative improvements and DIF-SASRec achieves $\textbf{12.42\%}$ and $\textbf{22.64\%}$ average relative improvements on three datasets in terms of Recall@$10$ and NDCG@$10$.
It demonstrates that the proposed solution can effectively fuse side information to help make next-item predictions for various attention-based SR models, indicating that it can potentially make more impact as a plug-in module in state-of-the-art SR models.

\subsection{Ablation and Hyper-parameter Study (RQ3)}

\subsubsection{Effectiveness of Different Components.}
\label{exp:components}
In order to figure out the contribution of different components in our proposed DIF-SR, we conduct ablation study for every proposed components on Sports and Yelp datasets, as shown in Tab.~\ref{tab: ablation}.:
\begin{itemize}
    \item (\textit{DIF-SR w/o DIF \& AAP}) DIF-SR without both DIF and AAP, which is identical to SASRec$_\text{F}$.
    \item (\textit{DIF-SR w/o DIF}) DIF-SR without DIF, which adopts early fusion as SASRec$_\text{F}$ and keep other settings including AAP same as the original model.
    \item (\textit{DIF-SR w/o AAP}) DIF-SR without AAP, which only uses the item predictor for training. 
\end{itemize}
Then we have the following observations:

First, DIF is the most effective component in DIF-SR framework. This could be demonstrated by the fact that \textit{DIF-SR w/o AAP} outperforms \textit{DIF-SR w/o DIF \& AAP} by a large margin. This observation verifies that decoupled side information significantly improves the representation power of attention matrices. 

Second, only using AAP on the prior integrated embedding based method \textbf{can not} always improve the performance, which is in line with our design: AAP is proposed to activate the attributes' influence on item-to-item attention in the self-attention layer, while prior solutions do not enable such influence.

Third, the AAP-based training paradigm can further improve the model performance when combined with DIF, verifying its capacity to activate the beneficial interaction and improve the performance.
%

\begin{table}[t]
\caption{Ablation study of DIF and AAP on Yelp, Sports and Beauty datasets.}
\resizebox{\linewidth}{!}{
\centering
\begin{tabular}{cc|cc|cc|cc}
\toprule
 \multicolumn{2}{c|}{\textbf{Settings}} & \multicolumn{2}{c|}{\textbf{Yelp}}& \multicolumn{2}{c|}{\textbf{Sports}}& \multicolumn{2}{c}{\textbf{Beauty}} \\
\textsc{DIF}  & \textsc{AAP}  & Recall@20 & +$\Delta$& Recall@20 & +$\Delta$& Recall@20 & +$\Delta$
\\ 
        \midrule
\xmark& \xmark & 0.0663 & - & 0.0621 & - & 0.0996 & -\\ 
\xmark& \cmark & 0.0663 & +0\% & 0.0754 & +21.42\% & 0.0991 & -0.50\% \\ 
\cmark& \xmark & 0.0968 & +46.00\% & 0.0767 & +23.51\% & 0.1240 & +24.50\%  \\ 
\cmark& \cmark & 0.1003  & +51.28\% & 0.0800 & +28.82\% & 0.1284 & +28.92\% \\ 
\bottomrule
\label{tab: ablation}
\end{tabular}}
\end{table}
\begin{table}[t]
\caption{Performance comparison of using different kinds of side information on Yelp dataset.}
    \resizebox{\linewidth}{!}{
\centering
\fontsize{9}{11}\selectfont 
\setlength{\tabcolsep}{0.4em}
\begin{tabular}{lcccc}
    \toprule
Side-info & Recall@10 & Recall@20 & NDCG@10 & NDCG@20 \\ \hline
Position & 0.0655 &  0.0954 & 0.0405 & 0.048 \\
Position + Categories & 0.0698 & 0.1003 & 0.0419 & 0.0496 \\
Position + City & 0.0691 & 0.1001 & 0.0415 & 0.0493 \\
Position + Categories + City & \textbf{0.0699} & \textbf{0.1010} & \textbf{0.0421} & \textbf{0.0499}  \\
    \bottomrule
\end{tabular}}
\label{tab:sideInfoAblation}
\end{table}

\subsubsection{Effectiveness of Different Kinds of Side Information.}
\label{exp:side info}
In order to figure out the contribution of different kinds of side information in our proposed DIF-SR, we conduct a study for various attributes. The position is the special and basic side information to enable order-aware self-attention, and other item-related side information of Yelp includes city and category. As shown in Tab.~\ref{tab:sideInfoAblation}, both item-related attributes can help the prediction by a large margin, demonstrating the effectiveness of side information fusion. Also, the combination of two attributes can further improve the performance slightly, which shows that our proposed solution can jointly leverage useful information from various resources.
\begin{figure}
    \centering
    \includegraphics[width=8.5cm]{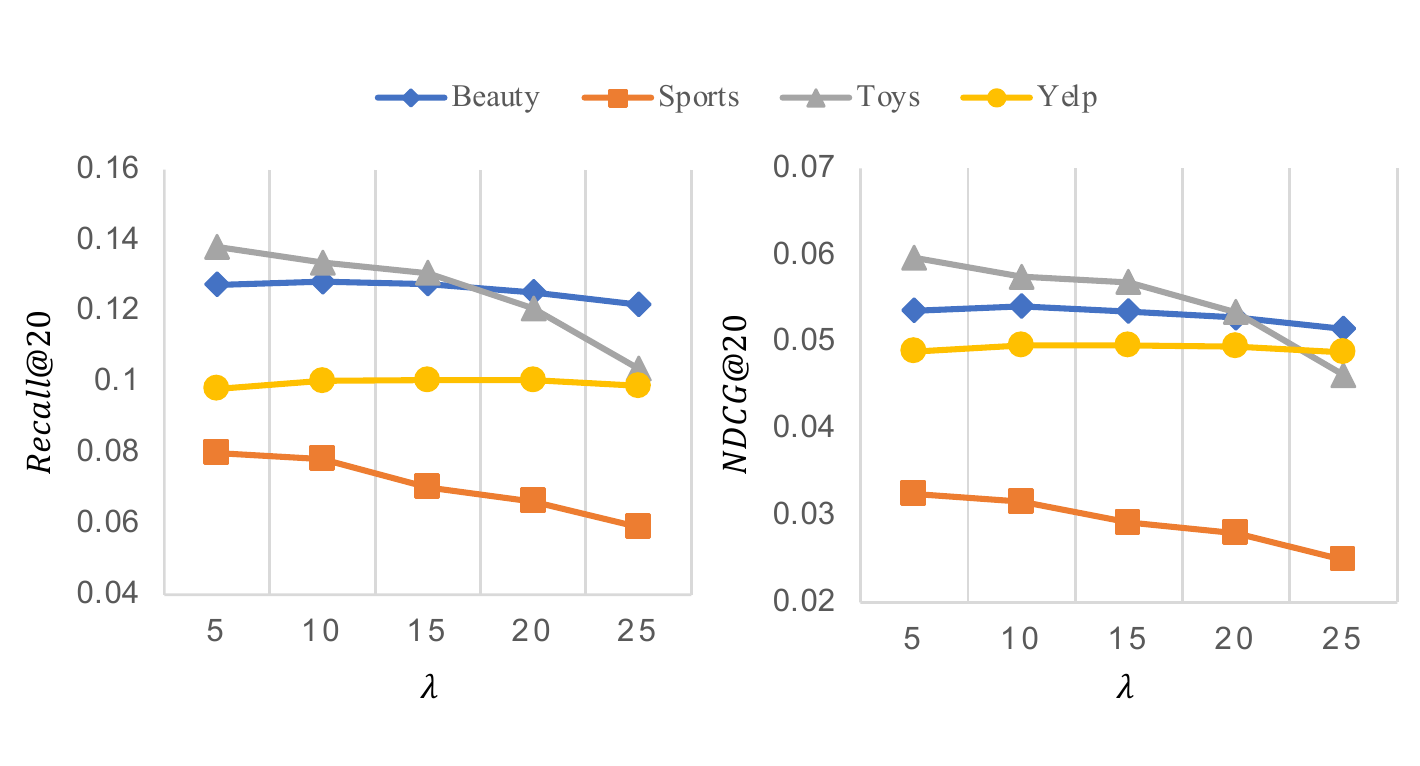}
    \caption{Effect of balance parameter $\lambda$.}
    \label{fig:param-lambda}
\end{figure}
\begin{figure}
    \centering
    \includegraphics[width=8.5cm]{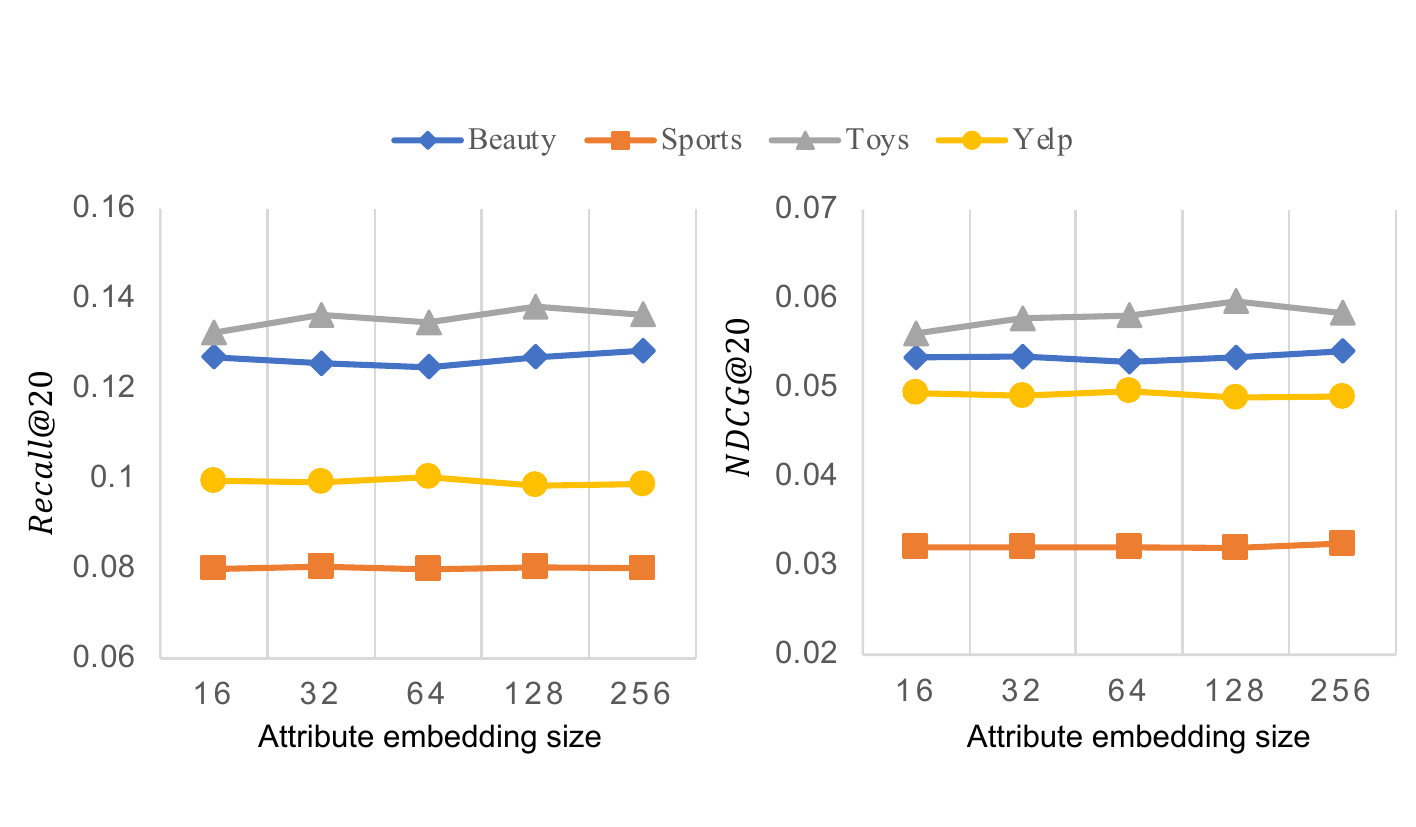}
    \caption{Effect of attribute embedding size $d_{f\cdot}$.}
    \label{fig:param-attribute}
\end{figure}

\subsubsection{Impact of Hyperparameters and Fusion Function.}
\label{exp:hyper}
In this experiment, we aim to investigate the effect of two hyperparameters, i.e., the balance parameter $\lambda $ for loss function and the attribute embedding size $d_{f}$ for category. We also experiment on the effect of fusion function $\mathcal{F}$.

Fig.~\ref{fig:param-lambda} presents the Recall@20 and NDCG@20 scores of DIF-SR with different $\lambda$. 
For all these four datasets, the best performance is achieved when the balance parameter $\lambda$ is 5 or 10. 
For Beauty and Yelp datasets, the model performance varies in extremely small range and   setting $\lambda$ to 10 is a better choice.

Fig.~\ref{fig:param-attribute} shows the effect of attribute embedding size on DIF-SR's performance. We observe that, in most cases, the performance is consistent across different choices of the attribute embedding size $d_{f}$ for category. Based on this, it is possible to greatly reduce the model complexity of DIF-SR by setting $d_f$ to a smaller value (typically smaller than the dimension of item embedding).

Fig.~\ref{fig:param-fusion} demonstrates that our method is also robust regarding different fusion functions, suggesting that simple fusion solutions, such as addition, do not harm the capacity of our model.



\subsection{Visualization of Attention Distribution (RQ4)}
To discuss the interpretability of DIF-SR, we visualize the attention matrices of test samples from the Yelp dataset. Due to the limitation of space, we show one example in Fig.~\ref{fig:visualization}. The two rows represent the attention matrices of the same head in different layers of one sample. The first three columns are decoupled attention matrices for the item and the attributes, and the last one is the fused attention matrices which are used to calculate the output item representation for each layer.

We derive the following observations from the results:
1) the decoupled attention matrices for different attributes show a different preference for capturing data patterns.
2) the fused attention matrix could adaptively adjust the contribution of each kind of side information via the decoupled attention matrices and synthesize crucial patterns from them.
\begin{figure}
    \centering
    \includegraphics[width=8.5cm]{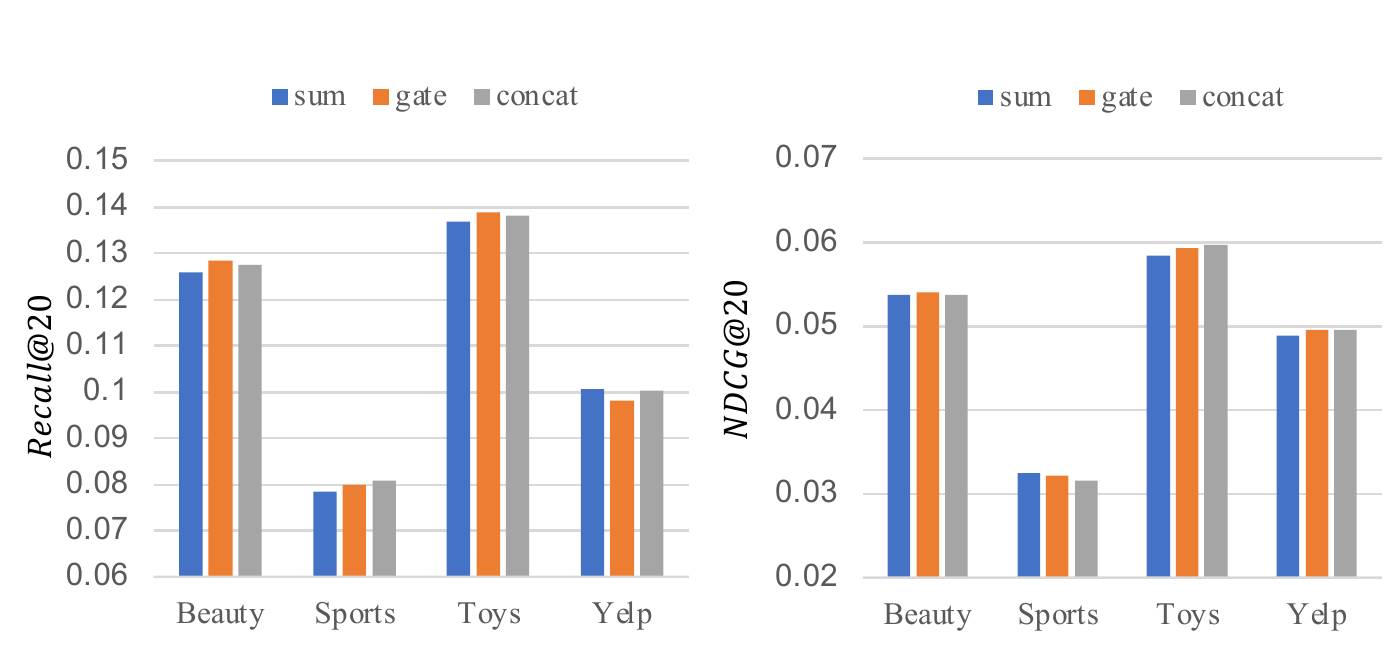}
    \caption{Effect of fusion functions $\mathcal{F}$.}
    \label{fig:param-fusion}
\end{figure}
\begin{figure}
    \centering
    \includegraphics[width=8.5cm]{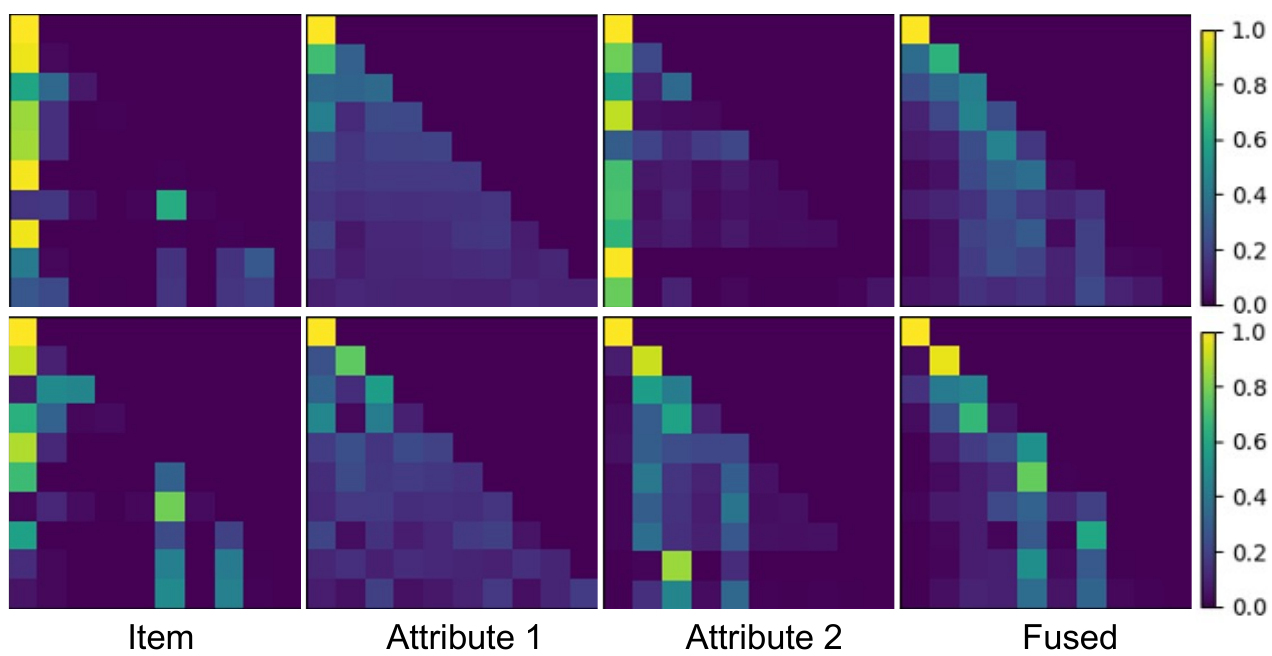}
    \caption{{Visualization of sampled attention matrices.}}
    \label{fig:visualization}
\end{figure}
\section{Conclusion}
We propose DIF-SR to effectively fuse side information for SR via moving side information from input to the attention layer, motivated by the observation that early integration of side information and item id in the input stage limits the representation power of attention matrices and flexibility of learning gradient. Specifically, we present a novel decoupled side information fusion attention mechanism, which allows higher rank attention matrices and adaptive gradient and thus enhances the learning of item representation. Auxiliary attribute predictors are also utilized upon the final representation in a multi-task training scheme to promote the interaction of side information and item representation. 

We provide a comprehensive theoretical analysis of the proposed DIF attention mechanism.  
Also, extensive studies on four public datasets show that our DIF-SR outperforms state-of-the-art basic SR models and competitive
side information integrated SR models. We also show that the DIF attention mechanism and AAP-based training scheme can be readily incorporated into attention-based SR models to benefit with side information fusion.
We hope that the future work can explore more possibilities of including side information in different stages of the model and further advance side information aware SR.

\section{Acknowledgement}
We thank RecBole team for a unified and comprehensive open-source recommendation library.
We thank Jaeboum Kim, Jipeng Zhang, Zhenxing Mi, Chenyang Qi, Junjie Wu, and Maosheng Ye for valuable discussion and suggestions. 

\newpage
\appendix
\section{Proofs} \label{appendix:proofs}
\begin{proof}[Proof of Theorem~\ref{thm:rank}]
For the first claim in Thm.~\ref{thm:rank}, note that we have $W_Q \in \mathbb{R}^{d \times d_h}$. Then, we can then leverage the well-proved theorem that the rank of the multiplication of the matrices is equal or less than the minimum of the individual ones: 
\begin{align*}
    rank(att) &= rank(R  W_Q  W_K^\top R^\top) \\
    &\le \min(rank(R), rank(W_Q), rank(R), rank(W_K)) \\
    &\le d_h,
\end{align*} 

For the second claim, we can construct a special case with the following strategy.
Let $W_Q = W_K$ be the matrices with rows $[0:d_h]$ as identity matrix. Let $W_Q^{f1} = W_K^{f1}$ be such with rows $[d_h:d_h + d_{h1}]$ as identity matrix. For $j > 1$, let $W_Q^{fj} = W_K^{fj}$ be such with rows $[d_h + \Sigma_{i = 1}^{j-1}d_{hi}: d_h + \Sigma_{i = 1}^{j}d_{hi}]$ as identity matrix. Other rows are kept all zeros.
Let $R^{(ID)}, E^{f1},\dots, E^{fp}$ be the matrices with rows $[0:d],[0:d_{f1}], \dots, [0:d_{fp}]$ as identity matrix and others as zeros.
Then we will get
\begin{align*}
    &DIF\_{att} = R^{(ID)} W_Q W_K^\top R^{(ID)\top} + \Sigma_{j=1}^{p}(E^{fj}W_Q^{fj}W_K^{fj\top}E^{fj\top}) \\
    =&\left( 
                      \begin{array}{cccc}
                             I_{d_h, d_h} &  & & \\
                             & I_{d_{h_1}, d_{h_1}}& & \\
                             & &\cdots& \\
                              & &I_{d_{h_p}, d_{h_p}}& \\
                              &  &  &\mathbf{0}_{n-(d_h+\Sigma_{j=1}^{p}d_{h_j}), n-(d_h+\Sigma_{j=1}^{p}d_{h_j})}
                      \end{array}
                \right),
\end{align*} 
where $I$ denotes the identity matrix, and empty as well as $\mathbf{0}$ denotes the all zeros matrix. 
Thus, we have
\begin{align*}
   rank(DIF\_{att}) = d_h + \Sigma_{j=1}^{p}d_{h_j} > d_h.
\end{align*}
\end{proof}

\begin{proof}[Proof of Theorem~\ref{thm:gradients}]
Let $\nabla_{E^{ID}} L$ and $\nabla_{E^{fj}} L$ denote the gradients of the objective function in Equ.\eqref{eqsas} for $E^{ID}$ and $E^{fj}, j \in [p]$, Then, we have
\begin{align*}
    \nabla_{E^{ID}} L &= \nabla_{G} L \cdot \nabla_{E^{ID} + \Sigma_{i=1}^{p}E^{fi}} G \cdot \nabla_{E^{ID}} (E^{ID} + \Sigma_{i=1}^{p}E^{fi}) \\
    & = \nabla_{G} L \cdot \nabla_{E^{ID} + \Sigma_{i=1}^{p}E^{fi}} G \\
    \nabla_{E^{fj}} L &= \nabla_{G} L \cdot \nabla_{E^{ID} + \Sigma_{i=1}^{p}E^{fi}} G \cdot \nabla_{E^{fj}} (E^{ID} + \Sigma_{i=1}^{p}E^{fi}) \\
    & = \nabla_{G} L \cdot \nabla_{E^{ID} + \Sigma_{i=1}^{p}E^{fi}} G.
\end{align*}
Similarly, let $\nabla_{E^{fj}} L$ denotes the gradient for the objective function in Equ.\eqref{eqnova} for $E^{fj}, j \in [p]$, then we have:
\begin{align*}
    \nabla_{E^{fj}} L &= \nabla_{G} L \cdot \nabla_{\Sigma_{i=1}^{p}E^{fi}} G \cdot \nabla_{E^{fj}} (\Sigma_{i=1}^{p}E^{fi}) \\
    & = \nabla_{G} L \cdot \nabla_{\Sigma_{i=1}^{p}E^{fi}} G,
\end{align*}
which is the same for every attribute $j$.
\end{proof}

\pagebreak

\bibliographystyle{ACM-Reference-Format} 
\bibliography{sample} 

\end{sloppypar}
\end{document}